\documentclass[10pt]{article}
\usepackage{amsfonts}
\usepackage{fullpage,graphicx,subfigure,mathdots,mathpazo,color}
\usepackage{amsmath,amscd,tikz,mathrsfs,cite}
\usepackage[normalem]{ulem}
\usepackage{amsmath}
\usepackage{setspace}
\usepackage{epsfig,amsmath,graphicx,amssymb,overpic}
\usepackage{bm}
\usepackage{graphicx}
\usepackage{subfigure, xcolor}
\usepackage{ntheorem}
\usepackage{listings,diagbox}
\usepackage{color}
\usepackage{epsfig,amsmath,graphicx,amssymb,overpic}
\usepackage{booktabs}
\usepackage{diagbox}
\usepackage{graphicx}
\usepackage{dcolumn}
\usepackage{bm}
\usepackage{graphicx}
\usepackage{subfigure}
\usepackage{epsfig,amsmath,graphicx,amssymb,overpic,cite}
\usepackage{makecell}
\usepackage{geometry}

\usepackage{epsfig,amsmath,graphicx,amssymb,overpic}
\usepackage{bm}
\usepackage{graphicx}
\usepackage{subfigure, xcolor}
\usepackage{ntheorem}

\def\be{\begin{equation}}
\def\ee{\end{equation}}
\def\bee{\begin{eqnarray}}
\def\ene{\end{eqnarray}}
\def\bes{\begin{subequations}}
\def\ees{\end{subequations}}
\def\no{\nonumber}

\newtheorem{prop}{Proposition}

\def\d{\displaystyle}

\def\v{\vspace{0.1in}}
\def\no{{\nonumber}}

\newtheorem*{proof}{Proof}
\newtheorem{RH}{Riemann-Hilbert problem}

\allowdisplaybreaks[4]

\newtheorem{remark}{Remark}
\usepackage{diagbox}

\begin{document}

\baselineskip=14pt \renewcommand {\thefootnote}{\dag}
\renewcommand
{\thefootnote}{\ddag} \renewcommand {\thefootnote}{ }

\pagestyle{plain}

\begin{center}
\baselineskip=16pt \leftline{} \vspace{-.3in} {\Large \bf  Breather gas and shielding of the focusing nonlinear Schr\"odinger
equation with nonzero backgrounds} \\[0.2in]
\end{center}

\begin{center}
Weifang Weng$^{1}$, Guoqiang Zhang$^{2}$, Boris A. Malomed$^{3
}$, Zhenya Yan$^{2,4,*}$\footnote{$^*${\it Email address}: zyyan@mmrc.iss.ac.cn (Corresponding author)}  \\[0.15in]
{\it \small $^{1}$School of Mathematical Sciences, University of Electronic Science and Technology of China, Chengdu, Sichuan, 611731, China \\
$^{2}$State Key Laboratory of Mathematical Sciences, Academy of Mathematics and Systems Science,\\  Chinese Academy of Sciences, Beijing 100190, China \\
$^{3}$Department of Physical Electronics, School of Electrical Engineering, Tel Aviv University,
Tel Aviv 69978, Israel\\
$^{4}$School of Mathematics and Information Science, Zhongyuan University of Technology, Zhengzhou 450007, China}
\end{center}

\baselineskip=13pt

\baselineskip=13pt

\vspace{0.2in}

\noindent {\bf Abstract:}\,
Breathers have been experimentally and theoretically found in many physical systems -- in particular, in  integrable nonlinear- wave models. A relevant problem is to study the \textit{breather gas}, which is the limit, for
$N\rightarrow \infty $, of $N$-breather solutions. In this paper, we investigate the breather gas in the framework of the focusing
nonlinear Schr\"{o}dinger (NLS) equation with nonzero boundary conditions, using
the inverse scattering transform and Riemann-Hilbert problem.
We address aggregate states in the form of $N$-breather solutions, when the
respective discrete spectra are concentrated in specific domains. We show
that the breather gas coagulates into a single-breather solution whose spectral eigenvalue is
located at the center of the circle domain, and a multi-breather solution
for the higher-degree quadrature concentration domain. These coagulation phenomena in the breather gas are called
\textit{breather shielding}. In particular, when the nonzero boundary
conditions vanish, the breather gas reduces to an $n$-soliton solution. When
the discrete eigenvalues are concentrated on a line, we derive the
corresponding Riemann-Hilbert problem. When the discrete spectrum is uniformly distributed within an ellipse, it is equivalent to the case of the line domain. These results may be useful to design experiments with breathers in physical settings.  \\



\vspace{0.1in}

\baselineskip=14pt

\vspace{0.1in}

\section{Introduction}

In 1834, solitary waves were discovered by Russell~\cite{sw}, and Korteweg and de Vries established the KdV equation to describe this wave phenomenon in 1895~\cite{kdv}. However, these significant results did not receive enough attention at that time. Until 1955, Fermi, Pasta and Ulam~\cite{FPU} numerically  investigated the thermalization process of a solid, which was called the Fermi-Pasta-Ulam (FPU) problem, and broke new branches of nonlinear science (e.g., solitons and chaos), and numerically simulating scientific problems~\cite{FPU-history}.
In 1965, Zabusky and Kruskal, motivated by the Fermi-Pasta-Ulam-Tsingou (FPUT) problem~\cite{FPU}, coined the concept of {\it `solitons'}, as elastically interacting solitary-waves solutions of the KdV equation (continuum limit of FPUT problem) with periodic initial data~\cite{soliton}. In 1967, Gardner {\it et al}~\cite{GGKM} discovered the inverse scattering transform (IST) to produce exact $\mathcal{N}$-soliton solutions of the KdV, starting from its spectral problem (alias the Lax pair~\cite{lax}), as
elaborated in detail in the classical work of Ablowitz {\it et al}~\cite{AKNS}. Parallel to that, the integrability of the
nonlinear Schr\"{o}dinger (NLS) equations and solitons produced by them were
discovered by Zakharov {\it et al} in 1971 \cite{zs72}. The predicted fundamental bright solitons and breathers (periodically oscillating $\mathcal{N}$th-order solitons,
which are also exact solutions of the NLS equation with the self-focusing
nonlinearity \cite{SY}) were created in optical fibers, for $%
\mathcal{N}=1$, $2$ and $3$ by Mollenauer {\it et al}~\cite{Mollenauer-1}. The $\mathcal{N}$th-order breather
may be considered as a bound state of $\mathcal{N}$ fundamental solitons
with unequal amplitudes, the ratios between which are $1:3:5:~\cdots~:2\mathcal{%
N}-1$. These bound states are fragile ones, in the sense that their binding
energy is zero~\cite{SY}. Nevertheless, they can be readily stabilized in fiber lasers~\cite{s2,s3,s4,s5},
where breathers are basic operation modes~\cite{laser}.

Another fundamentally important realization of the NLS equation (alias the Gross-Pitaevskii equation) is provided
by  quasi-one-dimensional Bose-Einstein condensates (BECs) with attractive
inter-atomic interactions~\cite{s9}. Fundamental solitons in BECs were first observed
in 2002 in condensates of $^{7}$Li atoms\ \cite{Randy,Khayk}. Breathers of
orders $2$ \cite{Oppo,we} and $3$ \cite{we} have been experimentally demonstrated
more recently. Moreover, solitons also appear in many fields of nonlinear science
~\cite{s6a,a7}.

In 1971, Zakharov first
proposed the concept of {\it soliton gas}, defined as the large-$N$ limit of the $N$-soliton
solution of the KdV~\cite{zak71}. It is relevant to stress that,
although collisions between solitons governed by integrable nonlinear wave equations are
elastic, collisional effects in the soliton gas are not trivial, as the
elastic collisions give rise to phase shifts of  solitons \cite%
{GGKM,zs72,AKNS,s1} (an exception is the 2D
KP-I equation, where collisions between weakly
localized lump solitons yield zero phase shifts \cite{KP}).
Afterwards, the concept was extended to investigate the fluid dynamics of
soliton gases, breather gases, and dense soliton gases for other nonlinear
wave equations~\cite{E1,E2,E3,E5,E6,sg1,sg2,sg3,sg4,sg5,sg6}. Especially, El \textit{et al}
elaborated the spectral theory~\cite{E4} and numerical experiment~\cite{E-num} of soliton and breather gases for the
NLS equation. Suret \textit{et al} developed the nonlinear spectral synthesis of the soliton gas in deep-water surface
gravity waves~\cite{E7}. There were some related soliton gas experiments in
optics~\cite{sg-exp2009,sg-exp2007,sg-exp2014} and  shallow water regime~\cite{sg-exp2019}.
In particular, the concept of soliton and breather gases is relevant for the implementation in fiber lasers, where
it is possible to create chains composed of large numbers of solitons and breathers \cite{sg5}.
Recently, Girotti {\it et al} first presented
the soliton gas of the KdV and modified KdV equations, respectively,
starting from $N$-soliton solutions via Riemann-Hilbert (RH)
problems~\cite{Girotti-1,Girotti-2}. Bertola {\it et al}
further proposed the effect of \textit{soliton shielding, }%
alias \textquotedblleft soliton coagulation", in dense soliton
gases governed by the NLS with zero
backgrounds~\cite{Grava-3,Bertola-pa}. The effect implies that the
field generated by a superposition of a large set of specially placed
solitons may become tantamount to a few-soliton configuration. However, the
\textquotedblleft coagulation" was not  studied for large sets of
NLS breathers, rather than fundamental solitons via RH problems. Compared to zero boundary condition of NLS equation, the discrete spectrum with nonzero boundary condition exhibit more symmetry. It is worth studying whether there is a phenomenon of breather-shielding effect in this case.

In this paper, motivated by Ref.~\cite{Grava-3} for the soliton gas of the NLS equation with zero
backgrounds, we would like to analyze the breather-shielding effect and breather gas
(the $N\rightarrow \infty $ limit of the $N$-breather
solutions) of the focusing NLS equation with nonzero boundary conditions (BCs) of the Dirichlet
type at infinity~\cite{zs72}:
\begin{equation}
\left\{
\begin{array}{l}
iq_{t}+q_{xx}+2( \left\vert q\right\vert ^{2}-q_{0}^{2})
q=0,\quad (x,t)\in \mathbb{R}^{2},\vspace{0.1in} \\
\displaystyle\lim_{x\rightarrow \pm \infty }q(x,t)=q_{\pm }=\mathrm{const}%
,\quad |q_{\pm }|=q_{0}>0.%
\end{array}%
\right.  \label{nls}
\end{equation}%
Our starting point is the Zakharov-Shabat (ZS) scattering
problem (i.e., the Lax pair)~\cite{zs72,Fa,Biondini2014}
\bee\label{lax-x}
\left\{\begin{array}{l}
\Psi _{x}=U\Psi ,\quad U=\dfrac{i}{2}\left( k-\dfrac{q_{0}^{2}}{k}\right)\sigma _{3}+Q,\vspace{0.05in} \\
\Psi _{t}=W\Psi ,\quad W=-\dfrac{i}{2}\left(k^2+\dfrac{q_{0}^4}{k^2}\right)\sigma _{3}\!+\! i\sigma_{3}
\left( Q_{x}\!-\!Q^{2}\right),
\end{array}\right.
\ene
where $\Psi =\Psi (x,t;k)$ is a second-order matrix-valued Jost function, $k$
is a complex spectral parameter,
and the potential function matrix
and $\sigma _{3}$ are given by
\begin{gather}\label{Q}
Q=\begin{bmatrix} 0& q(x,t) \vspace{0.05in}\\ -q^*(x,t)&0 \end{bmatrix},\qquad
\sigma_3=\begin{bmatrix}
1&0\vspace{0.05in}\\
0&-1
\end{bmatrix},
\end{gather}
with $\ast $ denoting the complex conjugate.
Notice that Eq.~(\ref{nls}) is the compatibility condition (or zero-curvature equation) $U_t-W_x+[U,\, W]=0$ of the Lax pair (\ref{lax-x}).

Different types (Kuznetsov-Ma-type~\cite{Ma1,Ma2}, Akhmediev-type~\cite{Ak}) of breathers were found for the NLS equation. Moreover, their parameter limits could generate its new non-periodic rogue wave (RW)~\cite{rw}.
Recently, a
powerful approach was proposed for applying IST and obtaining exact
solutions to the focusing and defocusing NLS equations with the nonzero BCs,
in terms of RH problems and their extensions, including the discrete
case, multi-component systems, and nonlocal equations~\cite{Ablowitz2007,
Biondini2014a, Kraus2015, Biondini2015, Biondini2016a, Prinari2018}. Deift {\it et al}
proposed the steepest-descent approximation for RH problems to
explore the long-time asymptotics of the modified KdV equation~\cite{RH}.
Techniques based on IST and RH problems were developed in other directions
\cite{RH2,n1, n1a, n2, n3,dbar1,dbar2,dbar3,dbar4,dbar5}. In particular, Bilman \textit{et al}
combined the robust IST and Darboux transform to obtain
RW solutions of the NLS equation~\cite{RIST}. The asymptotics of multi-soliton solutions to the NLS
equation was addressed~\cite{bil1,bil3}. Later, a scale
transform and RH technique were applied to an $N$-RW solution of the NLS equation
to analyze its near- and far-field asymptotic behaviors~\cite{bil2}.
Recently, Romero-Ros {\it et al} experimentally demonstrated the RW dynamics
in a 3D coupled BECs~\cite{bec-rw}.
Recently, Falqui {\it et al}~\cite{Falqui-phd} reported some results about the shielding of breathers of the NLS equation.
In fact, we independently finished our paper and presented the detailed
analysis and more examples about breather gas and shielding of the focusing NLS equation. We now summarize the main results of our work as follows:
\begin{itemize}
 \item {} We show that the breather gas condenses into a single-breather solution whose discrete spectrum is centered at the circle domain, and a multi-breather solution for higher-degree quadrature domains;
 \item {} When discrete spectra concentrate along a line segment, we derive the corresponding Riemann-Hilbert problem. When the discrete spectrum is uniformly distributed within an ellipse, it is equivalent to the line-segment domain case.
 \end{itemize}

The rest of this paper is arranged as follows. In Sec. 2, we simply recall the direct and inverse scattering transforms and the corresponding RH problem of the NLS equation with nonzero backgrounds. In Sec. 3, we analyze the breather gas, which is the limit of the $N$-breather solution at $N\to\infty$, via the modified RH problem. Moreover, we address aggregate states in the form of $N$-breather solutions, when the respective discrete spectra are concentrated in specific domains. We show
that the breather gas coagulates into a single-breather solution whose spectral eigenvalue is
located at the center of the domain for the circle domain, and a multi-breather solution
for the higher-degree quadrature concentration domain. These coagulation phenomena in the breather gas are called
\textit{breather shielding}. In particular, when the nonzero boundary
conditions vanish, the breather gas reduces to an $n$-soliton solution. When
the discrete eigenvalues are concentrated on a line, we derive the
corresponding Riemann-Hilbert problem. When the discrete spectrum is uniformly distributed within an ellipse, it is equivalent to the case of the line domain. Finally, we present some conclusions and discussions in Sec. 4.

\section{Preliminaries}

In this section, we recall the basic properties about the direct and inverse scattering problems and RH problem of the NLS equation with nonzero BCs given by Eq.~(\ref{nls}) (see \cite{Biondini2014} for the details). As $x\to \pm \infty$, the ZS scattering problem (or Lax pair) (\ref{lax-x}) becomes the asymptotic form
\bee\label{lax-x-2}
\left\{\begin{array}{l}
\Psi^{bg}_{x}=U^{bg}\Psi^{bg},\quad U^{bg}=\dfrac{i}{2}\left( k-\dfrac{q_{0}^{2}}{k}\right)\sigma _{3}+Q_{\pm},\vspace{0.05in} \\
\Psi^{bg}_{t}=W^{bg}\Psi^{bg},\quad W^{bg}=-\dfrac{i}{2}\left(k^2+\dfrac{q_{0}^4}{k^2}\right)\sigma _{3}\!-\! iq_0^2\sigma_{3},
\end{array}\right.
\ene
which admits the solution
\begin{align}
\Psi_{\pm}^{bg}(x, t; k)=
\left\{
\begin{alignedat}{2}
&P_{\pm}(k)\,\mathrm{e}^{i\vartheta(x, t; k)\sigma_3}, && k\not=0,\, \pm iq_0,\\[0.04in]
&\mathbb{I}_2+\left(x-2k\,t\right)(ik\sigma_3+Q_{\pm}),&\quad& k=\pm iq_0,
\end{alignedat}\right.
\end{align}
where $Q_{\pm}=\displaystyle\lim_{x\to\pm \infty}Q(x,t)=Q(x,t)\big|_{q=q_{\pm}}$ and
\begin{align}\label{ez}
P_{\pm}\left(k\right)=\mathbb{I}_2+\dfrac{i}{k}\sigma_3Q_{\pm},\quad
\vartheta\left(x, t; k\right)=\frac{1}{2}\left(k+\frac{q_0^2}{k}\right)\left[x-\left(k-\frac{q_0^2}{k}\right)t\right],
\end{align}
with $\mathbb{I}_2$ is a $2\times2$ unit matrix.

Let $\Sigma=\mathbb{R}\cup C_0,\, \widehat\Sigma=\mathbb{R}\backslash\{0\}\cup C_0,\, \Sigma_0:=\Sigma\backslash\{\pm iq_0\}$ with $C_0=\{k\in\mathbb{C}: |k|=q_0\}$. The continuous spectrum of $X_{\pm}=\lim_{x\to\pm\infty}X$ is the set of all values of $k$ satisfying $k+q_0^2/k\in\mathbb{R}$, i.e., $k\in\Sigma=\mathbb{R}\cup C_0$, which are the jump contours for the related Riemann-Hilbert problem (see the inverse scattering problem). Let $\mathbb{D}^{+}\equiv \{k|~\mathrm{Im}(k)(1-q_{0}^{2}/|k|^{2})>0\},\, \mathbb{D}^{-}\equiv \{k|~\mathrm{Im}%
(k)(1-q_{0}^{2}/|k|^{2})<0\}$. Thus, one can simultaneously determine the Jost and  modified Jost solutions $\Psi_{\pm}(x, t; k)$ and $\mu_{\pm}(x, t; k)$ of the Lax pair (\ref{lax-x}) satisfying the boundary conditions
\bee\label{bianjie0}
\begin{array}{l}
\Psi_{\pm}(x, t; k)=P_{\pm}(k)\,\mathrm{e}^{i\vartheta(x, t; k)\sigma_3}+o\left(1\right), \quad x\to\pm\infty, \vspace{0.1in}\\
\mu_{\pm}(x, t; k)=\Psi_{\pm}(x, t; k)\,\mathrm{e}^{-i\vartheta(x, t; k)\sigma_3}\to P_{\pm}(k),\quad x\to\pm\infty,
\end{array}
\ene
where
\begin{align}\label{Jost-int}
\mu_{\pm}(k)\!=\!\left\{
\begin{aligned}
&P_{\pm}(k)\!\left\{\mathbb{I}_2\!+\!\!\int_{\pm\infty}^x\!\!\exp\left(\frac{i}{2}\left(k+\frac{q_0^2}{k}\right)(x-\xi)
\widehat\sigma_3\right)\!\!\left[P_{\pm}^{-1}(k)[Q(\xi, t)-Q_{\pm}]\,\mu_{\pm}(\xi, t; k)\right]d\xi\right\}, \\
 & \qquad\qquad\qquad\qquad\qquad\qquad\qquad k\ne 0,\,\pm iq_0, \, q-q_{\pm}\in L^1\!\left(\mathbb{R^{\pm}}\right),\\[0.05in]
&P_{\pm}(k)\!+\!\!\int_{\pm\infty}^x\!\!\!\left[I\!+\!\left(x\!-\!\xi\right)\!\left(Q_{\pm}\mp q_0\,\sigma_3\right)\right]\!
[Q(\xi, t)-Q_{\pm}]\,\mu_{\pm}(\xi, t; k)d\xi, \\
& \qquad\qquad\qquad\qquad\qquad\qquad\qquad  k=\pm iq_0,\,\left(1\!+\!\left|x\right|\right)\!\left(q\!-\!q_{\pm}\right)\in\! L^1\!\!\left(\mathbb{R^{\pm}}\right)
\end{aligned}\right.
\end{align}
with $e^{\alpha\widehat \sigma_3}A:=e^{\alpha\sigma_3}Ae^{-\alpha\sigma_3}$.

Let $\Psi_{\pm}(x, t; k)=(\Psi_{\pm 1}(x, t; k),\, \Psi_{\pm 2}(x, t; k))$ and $\mu_{\pm}(x, t; k)=(\mu_{\pm 1}(x, t; k),\, \mu_{\pm 2}(x, t; k))$. Then for the given $q-q_{\pm}\in L^1\!\left(\mathbb{R^{\pm}}\right)$, the Jost functions $\Psi_{\pm 2}(x, t; k)$ and modified forms $\mu_{\pm 2}(x, t; k)$ given by Eqs.~(\ref{bianjie0}) and (\ref{Jost-int}) both possess the unique solutions in $\Sigma_0$.  Moreover, $\mu_{+1,-2}(x, t; k)$ and $\Psi_{+1, -2}(x, t; k)$ ($\mu_{-1,+2}(x, t; k)$  and  $\Psi_{-1,+2}(x, t; k)$) can be extended analytically to $\mathbb{D}^{+}$ ( $\mathbb{D}^{-}$), and continuously to $\mathbb{D}^{+}\cup\Sigma_0$ ($\mathbb{D}^{-}\cup\Sigma_0$). Since $\Psi_{\pm}(x, t; k),\, k\not=0,\,\pm iq_0$ are both fundamental matrix solutions of the Lax pair (\ref{lax-x}), thus there exists a constant scattering matrix $S(k)$ between them
 \bee\label{sr}
\Psi_+(x, t; k)=\Psi_-(x, t; k)\,S(k),\quad \mu_+(x, t; k)=\mu_-(x, t; k)e^{i\vartheta\hat\sigma_3}\,S(k),\quad  k\in\Sigma_0,
 \ene
where  $S(k)=\left(s_{ij}(k)\right)_{2\times 2}$ with the scattering coefficients $s_{ij}(k)$'s with
$s_{11}(k)$ and $s_{22}(k)$ in $k\in\Sigma_0$ being extended analytically to $\mathbb{D}^{+}$ and $\mathbb{D}^{-}$,
and continuously to $\mathbb{D}^{+}\cup\Sigma_0$ and $\mathbb{D}^{-}\cup\Sigma_0$, respectively, and
$s_{12}(k)$ and $s_{21}(k)$ not being analytically continued away from $\Sigma_0$. The reflection coefficients are defined as $
\rho(k)=\frac{s_{21}(k)}{s_{11}(k)},\,\widehat\rho(k)=\frac{s_{12}(k)}{s_{22}(k)}, \, k\in \Sigma_0.$
The Jost solutions $\Psi(x,t;z),\, \mu_{\pm}(x,t;z)$, the scattering matrix $S(z)$, and reflection coefficients admit the following symmetries
\bee\label{ps}
\begin{array}{l}
\Psi_{\pm}(x,t; k)=\sigma_2\Psi_{\pm}^*(x,t; k^*)\sigma_2=\dfrac{i}{k}\Psi_{\pm}\left(x,t; -\dfrac{q_0^2}{k}\right)\sigma_3Q_{\pm}, \v\\
\mu_{\pm}(x,t; k)=\sigma_2\mu_{\pm}^*(x,t; k^*)\sigma_2=\dfrac{i}{k}\mu_{\pm}\left(x,t; -\dfrac{q_0^2}{k}\right)\sigma_3Q_{\pm},\quad \sigma_2={\rm antidiag}(-i, i),  \v\\
S(k)=\sigma_2S^*(k^*)\sigma_2=(\sigma_3Q_-)^{-1}S\left(-\dfrac{q_0^2}{k}\right)\sigma_3Q_+, \quad
 \rho(k)=-\widehat{\rho}^*(k^*)=\dfrac{q_-^*}{q_-}\,\widehat{\rho}\left(-\dfrac{q_0^2}{k}\right).
\end{array}
  \ene

Moreover, the asymptotic behaviors for modified Jost solutions and scattering matrix are found by
\bee \label{mu-s}
\left\{\begin{array}{lll}\label{jianjin}
\mu_{\pm}(x, t; k)=\mathbb{I}_2+O\left(\dfrac{1}{k}\right), & S(k)=\mathbb{I}_2+O\left(\dfrac{1}{k}\right), \qquad & k\to\infty,\v \\
\mu_{\pm}(x, t; k)=\dfrac{i}{k}\,\sigma_3\,Q_{\pm}+O\left(1\right),& S(k)={\rm diag}\left(\dfrac{q_-}{q_+},\, \dfrac{q_+}{q_-}\right)+O\left(k\right), \qquad & k\to0.
\end{array}\right.
\ene
The discrete spectrum of the focusing NLS equation with nonzero BCs (\ref{nls}) is the set (see Figure \ref{fig2})
\bee
K=K^+\cup K^-,\quad K^+=\left\{k_j,\, -q_0^2/k_j^*\right\}_{j=1}^{N}\subset \mathbb{D}^+,\quad
 K^-=\left\{\ k_j^*, \, -q_0^2/k_j\right\}_{j=1}^{N}\subset \mathbb{D}^-,
\ene
where $s_{11}(k_j)=0,\, s_{11}'(k_j)\not=0$.

\begin{figure}[!t]
\centering
\vspace{-0.00in}\hspace{-0.35in} {\scalebox{0.4}[0.4]{\includegraphics{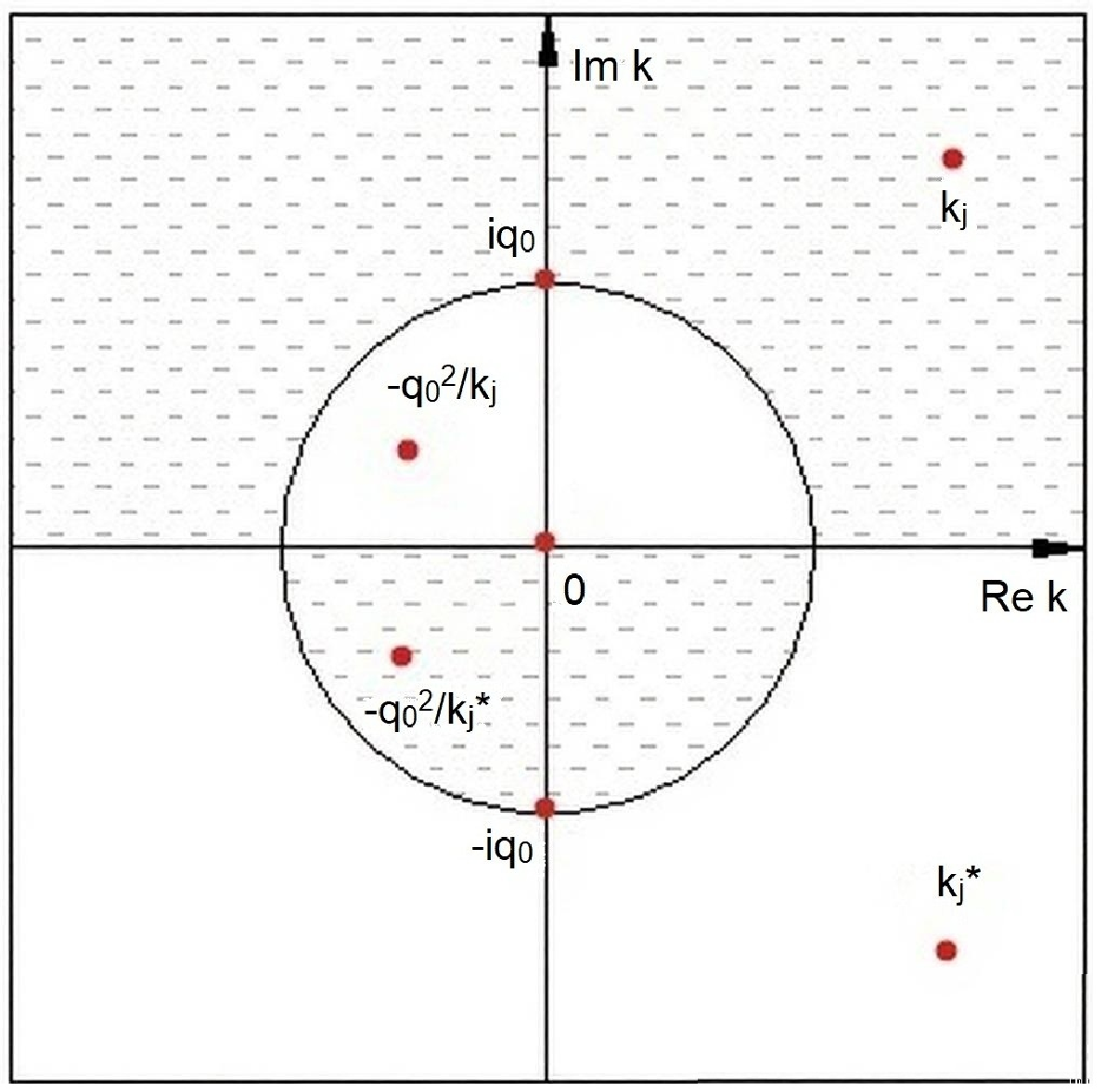}}}\hspace{-0.35in} \vspace{0.05in}
\caption{
The complex $k$-plane, showing the discrete spectrum $\{k_j,\, k_j^*,\, -q_0^2/k_j,\, -q_0^2/k_j^*\}_{j=1}^N$, and the shaded area indicates region $\mathbb{D}^+$, and the white area indicates region $\mathbb{D}^-$.}
\label{fig2}
\end{figure}

We construct a piecewise meromorphic function:
\bee\label{RHP-M}
M(x, t; k)=\left\{
\begin{array}{ll}
\left(\dfrac{\mu_{+1}(x, t; k)}{s_{11}(k)},\, \mu_{-2}(x, t; k)\right)
=\left(\dfrac{\Psi_{+1}(x, t; k)}{s_{11}(k)}, \Psi_{-2}(x, t; k)\right)\mathrm{e}^{-i\vartheta(x, t; k)\sigma_3}, & k\in \mathbb{D}^+, \v\v \\
\left(\mu_{-1}(x, t; k),\, \dfrac{\mu_{+2}(x, t; k)}{s_{22}(k)}\right)
=\left(\Psi_{-1}(x, t; k), \dfrac{\Psi_{+2}(x, t; k)}{s_{22}(k)}\right)\mathrm{e}^{-i\vartheta(x, t; k)\sigma_3},
  & k\in \mathbb{D}^-.
\end{array}\right.
\ene

Then, the matrix function $M(x,t;k)$ satisfies the following RH problem:

\begin{RH}\label{RH1}
Find a $2\times 2$ matrix $M(x,t;k)$ that satisfies the following conditions:

\begin{itemize}

 \item {} Analyticity: $M(x,t;k)$ is meromorphic in $\{k|k\in\mathbb{C}\setminus(\Sigma \cup K)\}$ and takes
continuous boundary values on $\Sigma$;

 \item {} The jump condition: the boundary values on the jump contour $\Sigma$ are defined as
 \bee
  M_+(k)=M_-(k)J(k), \quad J(k)=\mathrm{e}^{i\vartheta(x, t; k)\widehat\sigma_3}\left(\begin{array}{cc} 1-\rho(k)\widehat\rho(k)& -\widehat\rho(k)\\[0.05in] \rho(k)& 1 \end{array}\right),\,\,\, k\in\Sigma,
  \ene

 \item {} Normalization:
\bee \label{asy}
 M(x, t; k)=\left\{\begin{array}{ll}
    \mathbb{I}_2+O\left(1/k\right)  & k\to\infty, \v\\
    \dfrac{i}{k}\,\sigma_3\,Q_-+O\left(1\right), & k\to 0.
\end{array}\right.
\ene

\item {} The residue conditions: $M(x,t;k)$ has simple poles at each point in $K:=\{k_j,-\frac{q_0^2}{k_j^*},k_j^*,-\frac{q_0^2}{k_j}\}_{j=1}^N$ with
\begin{align}\label{M-res}
\begin{aligned}
&\mathop{\mathrm{Res}}\limits_{k=k_j}M(x,t;k)=\lim\limits_{k\to k_j}M(x,t;k)\left[\!\!\begin{array}{cc}
0& 0  \vspace{0.05in}\\
c_je^{-2i\vartheta(k_j)}& 0
\end{array}\!\!\right],\vspace{0.05in}\\
&\mathop{\mathrm{Res}}\limits_{k=-\frac{q_0^2}{k_j^*}}M(x,t;k)=\lim\limits_{k\to -\frac{q_0^2}{k_j^*}}M(x,t;k)\left[\!\!\begin{array}{cc}
0& 0  \vspace{0.05in}\\
-\frac{q_0^2q_-^*}{k_j^{*2}q_-}c_j^*e^{-2i\vartheta(-\frac{q_0^2}{k_j^*})}& 0
\end{array}\!\!\right],\vspace{0.05in}\\
&\mathop{\mathrm{Res}}\limits_{k=k_j^*}M(x,t;k)=\lim\limits_{k\to k_j^*}M(x,t;k)\left[\!\!\begin{array}{cc}
0& -c_j^*e^{2i\vartheta(k_j^*)}  \vspace{0.05in}\\
0& 0
\end{array}\!\!\right], \vspace{0.05in}\\
&\mathop{\mathrm{Res}}\limits_{k=-\frac{q_0^2}{k_j}}M(x,t;k)=\lim\limits_{k\to -\frac{q_0^2}{k_j}}M(x,t;k)\left[\!\!\begin{array}{cc}
0& \frac{q_0^2q_-}{k_j^2q_-^*}c_je^{2i\vartheta(-\frac{q_0^2}{k_j})}  \vspace{0.05in}\\
0& 0
\end{array}\!\!\right],
\end{aligned}
\end{align}
with $c_j$'s being complex constants(see Table \ref{table:s1}).
\end{itemize}
\end{RH}

\begin{table}[!hbt]
\renewcommand{\arraystretch}{1.5}
\begin{center}
\caption{The relationships between the norming constants.}
\begin{tabular}{c|c|c|c|c|c}
  \hline
  \hline
  \diagbox{$K$}{$j$} & 1 & 2 & 3 & 4 &$\cdots$ \\
  \hline
  $k_j$ & $c_1$ & $c_2$ & $c_3$ & $c_4$ & \quad\\[1em]
  $-\frac{q_0^2}{k_j^*}$ & $-\frac{q_0^2q_-^*}{k_1^{*2}q_-}c_1^*$ & $-\frac{q_0^2q_-^*}{k_2^{*2}q_-}c_2^*$ & $-\frac{q_0^2q_-^*}{k_3^{*2}q_-}c_3^*$ &$-\frac{q_0^2q_-^*}{k_4^{*2}q_-}c_4^*$ & \quad \\[1em]
  $k_j^*$ & $-c_1^*$ & $-c_2^*$ & $-c_3^*$ & $-c_4^*$ & \quad\\[1em]
  $-\frac{q_0^2}{k_j}$ & $\frac{q_0^2q_-}{k_1^2q_-^*}c_1$ & $\frac{q_0^2q_-}{k_2^2q_-^*}c_2$ & $\frac{q_0^2q_-}{k_3^2q_-^*}c_3$ & $\frac{q_0^2q_-}{k_4^2q_-^*}c_4$ & \quad\\
  $\cdots$&$\cdots$&$\cdots$&$\cdots$&$\cdots$&$\cdots$\\
  \hline
  \hline
\end{tabular}
\end{center}
\label{table:s1}
\end{table}\par

Then for the reflectionless case $\rho(k)=0$,  the $N$-breather solution $q(x,t)$ of the focusing NLS equation with nonzero BCs is given by
\bee\label{fanyan}
q(x,t)=-i\lim\limits_{k\rightarrow\infty}\left(k M(x,t;k)\right)_{12},
\ene
where $M(x,t;k)$ is determined by Eqs.~(\ref{M-res}) and (\ref{asy}) as
\begin{equation}
M(x,t;k)=\mathbb{I}+\dfrac{i}{k}\,\sigma _{3}\,Q_{-}+\sum_{j=1}^{2N}\left( \frac{%
\left[ \!\!%
\begin{array}{cc}
u_{j}(x,t; \eta) & 0 \\
v_{j}(x,t; \eta) & 0%
\end{array}%
\!\!\right] }{k-\eta _{j}}+\frac{\left[ \!\!%
\begin{array}{cc}
0 & \widehat{v}_{j}(x,t; \widehat{\eta}) \\
0 & \widehat{u}_{j}(x,t; \widehat{\eta})%
\end{array}%
\!\!\right] }{k-\widehat{\eta}_{j}}\right) ,
\end{equation}%
with $u_{j}=u_{j}(x,t; \eta),\,v_{j}=v_{j}(x,t; \eta),\, \widehat{u}_{j}=\widehat{u}_{j}(x,t; \widehat{\eta}),\,\widehat{v}_{j}=\widehat{v}_{j}(x,t; \widehat{\eta})$ that can be found from Eq.~(\ref%
{M-res}), $Q_-=\mathrm{antidiag}(q_-,-q^{\ast}_-)$, and $\eta _{j}=k_{j},\,\eta _{N+j}=-q_{0}^{2}/k_{j}^{\ast
}, \widehat{\eta}_j=-q_0^2/\eta_j,\,  \widehat{\eta}_{N+j}=-q_0^2/\eta_{N+j},\,(j=1,2,...,N)$. Note that the limiting $M(k)$ satisfies, in general, a
$\bar\partial$ problem. In particular, the discrete spectra must satisfy the ``theta" condition~\cite{Biondini2014}:
\bee
\mathrm{arg} (\frac{q_-}{q_+})=4\sum\limits_{j=1}^N\mathrm{arg} k_j.
\ene

\section{Breather gas: the limit of the $N$-breather solution at $N\to\infty$}

Below, we consider the ZS spectral problem for a reflectionless potential ($\rho(k)=0,\, k\in\Sigma$) with simple poles, which corresponds to focusing NLS equation with nonzero BC. We define a smooth contour $\Gamma_{1+}(\Gamma_{2+})$ in the domain $\mathbb{D}_+$, oriented counterclockwise, that encircles all the poles $\{k_j\}_{j=1}^N(\{-\frac{q_0^2}{k_j^*}\}_{j=1}^N)$ , and a smooth contour $\Gamma_{1-}(\Gamma_{2-})$ in the domain $\mathbb{D}_-$, oriented clockwise, that encircles all the poles $\{k_j^*\}_{j=1}^N(\{-\frac{q_0^2}{k_j}\}_{j=1}^N)$.

Based on the RH problem~\ref{RH1}, we consider following transformation:
\bee
M_1(x,t;k)
=\begin{cases}
M(x,t;k)\left[\!\!\begin{array}{cc}
1& 0  \vspace{0.05in}\\
-\sum\limits_{j=1}^{N}\dfrac{c_je^{-2i\vartheta(k_j)}}{k-k_j}& 1
\end{array}\!\!\right],\quad k~\mathrm{within}~\Gamma_{1+},\vspace{0.05in}\\
M(x,t;k)\left[\!\!\begin{array}{cc}
1& 0  \vspace{0.05in}\\
\sum\limits_{j=1}^{N}\dfrac{\frac{q_0^2q_-^*}{k_j^{*2}q_-}c_j^*e^{-2i\vartheta(-\frac{q_0^2}{k_j^*})}}{k+\frac{q_0^2}{k_j^*}}& 1
\end{array}\!\!\right],\quad k~\mathrm{within}~\Gamma_{2+},\vspace{0.05in}\\
M(x,t;k)\left[\!\!\begin{array}{cc}
1& \sum\limits_{j=1}^{N}\dfrac{c_j^*e^{2i\vartheta(k_j^*)}}{k-k_j^*}  \vspace{0.05in}\\
0& 1
\end{array}\!\!\right],\quad k~\mathrm{within}~\Gamma_{1-},\v\\
M(x,t;k)\left[\!\!\begin{array}{cc}
1& -\sum\limits_{j=1}^{N}\dfrac{\frac{q_0^2q_-}{k_j^2q_-^*}c_je^{2i\vartheta(-\frac{q_0^2}{k_j})}}{k+\frac{q_0^2}{k_j}}  \vspace{0.05in}\\
0& 1
\end{array}\!\!\right],\quad k~\mathrm{within}~\Gamma_{2-},\v\\
M(x,t;k),\quad \mathrm{otherwise}.
\end{cases}
\ene
Therefore, according to RH problem~\ref{RH1}, we know that the matrix function $M_1(x,t;k)$ satisfies the following RH problem:

\begin{RH}\label{RH2}
Find a $2\times 2$ matrix function $M_1(x,t;k)$ that meets the following conditions:

\begin{itemize}

 \item {} Analyticity: $M_1(x,t;k)$ is analytic in $\mathbb{C}\setminus(\Gamma_{1\pm}\cup\Gamma_{2\pm})$ and takes continuous boundary values on $\Gamma_{1\pm}\cup\Gamma_{2\pm}$.

 \item {} The jump condition: The boundary values on the jump contour $\Gamma_{1+}\cup\Gamma_{1-}$ are defined as
 \bee
M_{1+}(x,t;k)=M_{1-}(x,t;k)V_1(x,t;k),\quad \lambda\in\Gamma_{1\pm}\cup\Gamma_{2\pm},
\ene
where
\bee\label{V1-1}
V_1(x,t;k)
=\begin{cases}
\left[\!\!\begin{array}{cc}
1& 0  \vspace{0.05in}\\
-\sum\limits_{j=1}^{N}\dfrac{c_je^{-2i\vartheta(k_j)}}{k-k_j}& 1
\end{array}\!\!\right],\quad k\in\Gamma_{1+},\vspace{0.05in}\\
\left[\!\!\begin{array}{cc}
1& 0  \vspace{0.05in}\\
\sum\limits_{j=1}^{N}\dfrac{\frac{q_0^2q_-^*}{k_j^{*2}q_-}c_j^*e^{-2i\vartheta(-\frac{q_0^2}{k_j^*})}}{k+\frac{q_0^2}{k_j^*}}& 1
\end{array}\!\!\right],\quad k\in\Gamma_{2+}, \vspace{0.05in}\\
\left[\!\!\begin{array}{cc}
1& -\sum\limits_{j=1}^{N}\dfrac{c_j^*e^{2i\vartheta(k_j^*)}}{k-k_j^*}  \vspace{0.05in}\\
0& 1
\end{array}\!\!\right],\quad k\in\Gamma_{1-},\v\\
\left[\!\!\begin{array}{cc}
1& \sum\limits_{j=1}^{N}\dfrac{\frac{q_0^2q_-}{k_j^2q_-^*}c_je^{2i\vartheta(-\frac{q_0^2}{k_j})}}{k+\frac{q_0^2}{k_j}}  \vspace{0.05in}\\
0& 1
\end{array}\!\!\right],\quad k\in\Gamma_{2-}.
\end{cases}
\ene

 \item {} Normalization:
\bee
 M_1(x, t; k)=\left\{\begin{array}{ll}
    \mathbb{I}_2+O\left(1/k\right)  & k\to\infty, \v\\
    \dfrac{i}{k}\,\sigma_3\,Q_-+O\left(1\right), & k\to 0.
\end{array}\right.
\ene

\end{itemize}
\end{RH}

According to Eq.~(\ref{fanyan}), we recover $q(x,t)$ by means of the following formula:
\bee\label{fanyan-1}
q(x,t)=-i\lim\limits_{k\rightarrow\infty}\left(k M_{1}(x,t;k)\right)_{12}.
\ene

Below, we address the limit of $N\to\infty$, under the additional assumptions:
\begin{itemize}
 \item {} The discrete spectra $k_j,j=1,\cdots,N$ with the norming constants $c_j,j=1,\cdots,N$ fill uniformly domain $\Omega_1$ which is strictly contained in the domain $D_{\Gamma_1}$ bounded by $\Gamma_1$ and domain $\Omega_1$ satisfies Green's theorem.

 \item {} The normalization constants $c_j,j=1,\cdots,N$ have the following form:
\bee\label{c-j}
c_j=\frac{|\Omega_1|r(k_j,k_j^*)}{N\pi}.
\ene
where $|\Omega_1|$ means the area of the domain $\Omega_1$ and $r(k,k^*):=n(k^*-s_1^*)^{n-1}r_1(k)$ is a smooth function of variables $k$ and $k^*$ with $s_1\in\mathbb{C}_+$ and smooth function $r_1(k)$ is
subject to the symmetry relation $r_1^*(k)=r_1(k^*)$. It should be noted that quadrature domains are employed to streamline the evaluation of contour integrals.

\item {} The discrete spectra satisfy the ``theta" condition:
\bee
\mathrm{arg} (\frac{q_-}{q_+})=4\sum\limits_{j=1}^N\mathrm{arg} k_j.
\ene
\end{itemize}

\begin{prop}\label{le1}
Let $(x,t)$ be in compact subsets of $\mathbb{R}^2$. For any open set $B_+$ containing the domain $\Omega_1$, the following identities hold:
\begin{align}
\begin{aligned}
&\lim\limits_{N\to\infty}\sum\limits_{j=1}^{N}\frac{c_je^{-2i\vartheta(k_j)}}{k-k_j}=\iint_{\Omega_1}
\frac{r(\lambda,\lambda^*)e^{-2i\vartheta(\lambda)}}{2\pi i(k-\lambda)}d\lambda^*\wedge d\lambda,\v\\
&\lim\limits_{N\to\infty}\sum\limits_{j=1}^{N}\dfrac{\frac{q_0^2q_-^*}{k_j^{*2}q_-}c_j^*e^{-2i\vartheta(-\frac{q_0^2}{k_j^*})}}
{k+\frac{q_0^2}{k_j^*}}
=\iint_{\Omega_1}\dfrac{\frac{q_0^2q_-^*}{\lambda^{*2}q_-}r^*(\lambda,\lambda^*)e^{-2i\vartheta(-\frac{q_0^2}{\lambda^*})}}
{2\pi i(k+\frac{q_0^2}{\lambda^*})}d\lambda^*\wedge d\lambda,
\end{aligned}
\end{align}
uniformly for all $\mathbb{C}\setminus B_+$. The boundary $\partial\Omega_1$ has the counterclockwise orientation.

\end{prop}

\begin{proof}
Using Eq.~(\ref{c-j}), we have
\begin{align}
\begin{aligned}
\lim\limits_{N\to\infty}\sum\limits_{j=1}^{N}\frac{c_je^{-2i\vartheta(k_j)}}{k-k_j}
=\lim\limits_{N\to\infty}\sum\limits_{j=1}^{N}\frac{|\Omega_1|}{N}\frac{r(k_j,k_j^*)e^{-2i\vartheta(k_j)}}{\pi(k-k_j)}
=\iint_{\Omega_1}\frac{r(\lambda,\lambda^*)e^{-2i\vartheta(\lambda)}}{2\pi i(k-\lambda)}d\lambda^*\wedge d\lambda.
\end{aligned}
\end{align}

Thus the proof is completed.
\end{proof}

It is noted that the Riemann-Hilbert problem with such jumps specified in Proposition \ref{le1} exists, as it is proved in \cite{Falqui-phd}.

\begin{prop}\label{le2}
The following identities hold:
\begin{align}
\begin{aligned}
&\iint_{\Omega_1}\frac{r(\lambda,\lambda^*)e^{-2i\vartheta(\lambda)}}{2\pi i(k-\lambda)}d\lambda^*\wedge d\lambda=
\int_{\partial{\Omega_1}}\frac{(\lambda^*-s_1^*)^{n}r_1(\lambda)e^{-2i\vartheta(\lambda)}}{2\pi i(k-\lambda)}d\lambda,\v\\
&\iint_{\Omega_1}\dfrac{\frac{q_0^2q_-^*}{\lambda^{*2}q_-}r^*(\lambda,\lambda^*)e^{-2i\vartheta(-\frac{q_0^2}{\lambda^*})}}
{2\pi i(k+\frac{q_0^2}{\lambda^*})}d\lambda^*\wedge d\lambda
=-\int_{\partial{\Omega_1}}\dfrac{\frac{q_0^2q_-^*}{\lambda^{*2}q_-}(\lambda-s_1)^{n}r_1^*(\lambda)e^{-2i\vartheta(-\frac{q_0^2}{\lambda^*})}}
{2\pi i(k+\frac{q_0^2}{\lambda^*})}d\lambda^*,
\end{aligned}
\end{align}
uniformly for all $\mathbb{C}\setminus \Omega_1$. The boundary $\partial\Omega_1$ has the counterclockwise orientation.

\end{prop}

\begin{proof}
Note that $r(k,k^*):=nk^{*(n-1)}r_1(k)$, using Green theorem, we have
\begin{align}\no
\begin{aligned}
\iint_{\Omega_1}\frac{r(\lambda,\lambda^*)e^{-2i\vartheta(\lambda)}}{2\pi i(k-\lambda)}d\lambda^*\wedge d\lambda
=&\iint_{\Omega_1}\frac{\overline{\partial}((\lambda^*-s_1^*)^{n})r_1(\lambda)e^{-2i\vartheta(\lambda)}}{2\pi i(k-\lambda)}d\lambda^*\wedge d\lambda\v\\
=&\iint_{\Omega_1}\overline{\partial}\left(\frac{(\lambda^*-s_1^*)^{n}r_1(\lambda)e^{-2i\vartheta(\lambda)}}{2\pi i(k-\lambda)}\right)d\lambda^*\wedge d\lambda\v\\
=&\int_{\partial{\Omega_1}}\frac{(\lambda^*-s_1^*)^{n}r_1(\lambda)e^{-2i\vartheta(\lambda)}}{2\pi i(k-\lambda)}d\lambda,
\end{aligned}
\end{align}
and
\begin{align}\no
\begin{aligned}
\iint_{\Omega_1}\dfrac{\frac{q_0^2q_-^*}{\lambda^{*2}q_-}r^*(\lambda,\lambda^*)e^{-2i\vartheta(-\frac{q_0^2}{\lambda^*})}}
{2\pi i(k+\frac{q_0^2}{\lambda^*})}d\lambda^*\wedge d\lambda
=&\iint_{\Omega_1}\dfrac{\frac{q_0^2q_-^*}{\lambda^{*2}q_-}\partial((\lambda-s_1)^{n})r_1^*(\lambda)e^{-2i\vartheta(-\frac{q_0^2}{\lambda^*})}}
{2\pi i(k+\frac{q_0^2}{\lambda^*})}d\lambda^*\wedge d\lambda\v\\
=&-\int_{\partial{\Omega_1}}\dfrac{\frac{q_0^2q_-^*}{\lambda^{*2}q_-}(\lambda-s_1)^{n}r_1^*(\lambda)e^{-2i\vartheta(-\frac{q_0^2}{\lambda^*})}}
{2\pi i(k+\frac{q_0^2}{\lambda^*})}d\lambda^*.
\end{aligned}
\end{align}

Thus the proof is completed.
\end{proof}

\subsection{Quadrature domains}

In this subsection, we consider the following special cases: The discrete spectra $k_j,j=1,\cdots,N$ fill uniformly domain $\Omega_1$ which is strictly contained in the domain $D_{\Gamma_1}$, that is,
\bee
\Omega_1:=\{k|~|(k-s_1)^m-s_2|<s_3\},\quad \Omega_1\subset D_+,
\ene
where $m\in\mathbb{N}_+$ and $|s_2|,s_3$ are sufficiently small (see Figure \ref{fig-1}(a)). Note that Bertola {\it et al}~\cite{Grava-3} have
discussed the soliton shielding of the NLS equation with zero BC in the domain
$\Omega _{1}$.

\begin{figure}[!t]
\centering
\vspace{-0.00in}\hspace{-0.35in} {\scalebox{0.45}[0.45]{\includegraphics{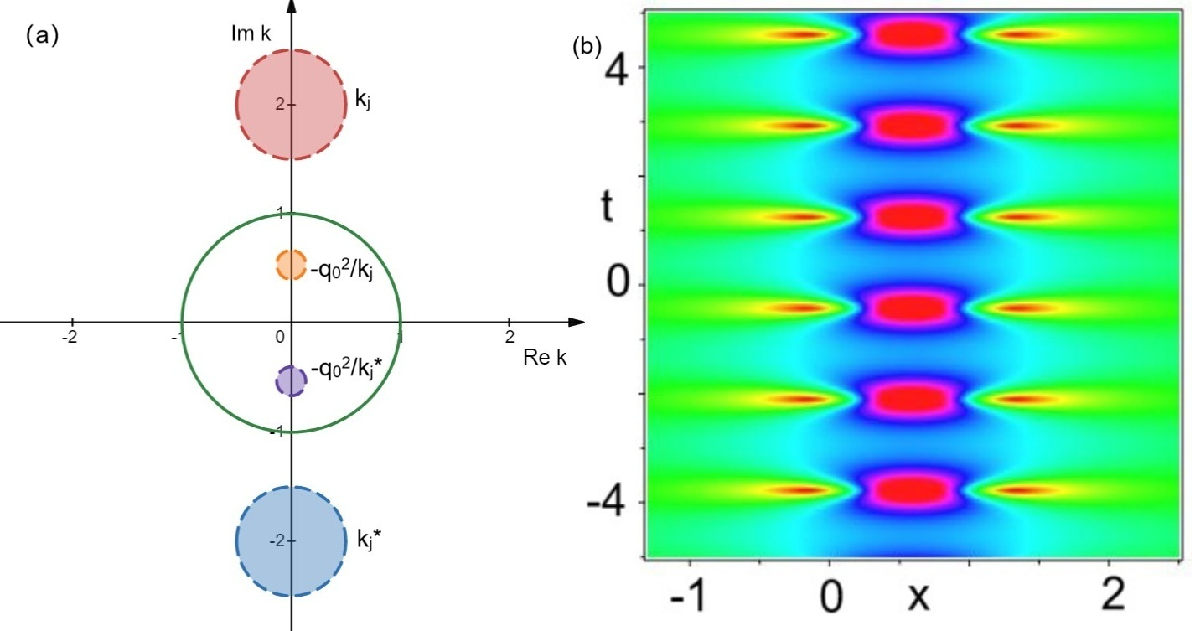}}}\hspace{-0.35in} \vspace{0.05in}
\caption{(a) The distribution of discrete spectrum $K$, the parameters being
$s_{1}=2i-\frac{1}{4},s_{2}=\frac{1}{4},s_{3}=\frac{1}{2},m=1,q_{0}=1$. (b) The
1-breather solution with the same parameters and $r_1(k)=4$.}
\label{fig-1}
\end{figure}

According to RH problems~\ref{le1} and \ref{le2}, we arrive at the following RH problem $M_2(x,t;k):=\lim\limits_{N\to\infty}M_1(x,t;k)$:

\begin{RH}\label{RH3}
Find a $2\times 2$ matrix function $M_2(x,t;k)$ that meets the following conditions:

\begin{itemize}

 \item {} Analyticity: $M_2(x,t;k)$ is analytic in $\mathbb{C}\setminus(\Gamma_{1\pm}\cup\Gamma_{2\pm})$ and takes continuous boundary values on $\Gamma_{1\pm}\cup\Gamma_{2\pm}$.

 \item {} The jump condition: The boundary values on the jump contour $\Gamma_{1+}\cup\Gamma_{1-}$ are defined as
 \bee
M_{2+}(x,t;k)=M_{2-}(x,t;k)V_2(x,t;k),\quad \lambda\in\Gamma_{1\pm}\cup\Gamma_{2\pm},
\ene
where
\bee\label{v2}
V_2(x,t;k)
=\begin{cases}
\left[\!\!\begin{array}{cc}
1& 0  \vspace{0.05in}\\
-\d\int_{\partial{\Omega_1}}\frac{(\lambda^*-s_1^*)^{n}r_1(\lambda)e^{-2i\vartheta(\lambda)}}{2\pi i(k-\lambda)}d\lambda& 1
\end{array}\!\!\right],\quad k\in\Gamma_{1+},\vspace{0.05in}\\
\left[\!\!\begin{array}{cc}
1& 0  \vspace{0.05in}\\
-\d\int_{\partial{\Omega_1}}\dfrac{\frac{q_0^2q_-^*}{\lambda^{*2}q_-}(\lambda-s_1)^{n}r_1^*(\lambda)
e^{-2i\vartheta(-\frac{q_0^2}{\lambda^*})}}
{2\pi i(k+\frac{q_0^2}{\lambda^*})}d\lambda^*& 1
\end{array}\!\!\right],\quad k\in\Gamma_{2+},\vspace{0.05in}\\
\left[\!\!\begin{array}{cc}
1& \d\int_{\partial{\Omega_1}}\frac{(\lambda-s_1)^{n}r_1^*(\lambda)e^{2i\vartheta(\lambda^*)}}{2\pi i(k-\lambda^*)}d\lambda^*  \vspace{0.05in}\\
0& 1
\end{array}\!\!\right],\quad k\in\Gamma_{1-},\v\\
\left[\!\!\begin{array}{cc}
1& \d\int_{\partial{\Omega_1}}\dfrac{\frac{q_0^2q_-}{\lambda^2q_-^*}(\lambda^*-s_1^*)^{n}r_1(\lambda)
e^{2i\vartheta(-\frac{q_0^2}{\lambda})}}{2\pi i(k+\frac{q_0^2}{\lambda})}d\lambda  \vspace{0.05in}\\
0& 1
\end{array}\!\!\right],\quad k\in\Gamma_{2-}.
\end{cases}
\ene

 \item {} The normalization:
\bee
 M_2(x, t; k)=\left\{\begin{array}{ll}
    \mathbb{I}_2+O\left(1/k\right)  & k\to\infty, \v\\
    \dfrac{i}{k}\,\sigma_3\,Q_-+O\left(1\right), & k\to 0.
\end{array}\right.
\ene

\end{itemize}
\end{RH}

According to Eq.~(\ref{fanyan-1}), we recover $q(x,t)$ by means of the following formula:
\bee\label{fanyan-2}
q(x,t)=-i\lim\limits_{k\rightarrow\infty}\left(k M_{2}(x,t;k)\right)_{12}.
\ene

To restrict the values of $n$ and $m$, we analyze the following three situations.

\v \noindent {\bf Case I.} The single-breather solution. In this case, we choose $n=m=1$. Then we arrive at the following Proposition.

\begin{prop}\label{prop3} Let $\lambda_0:=s_1+s_2$, then the solution of the RH problem \ref{RH3} is a single-breather solution $q_1(x,t)$ with the discrete eigenvalue $\lambda_0$ and normalization constants $c_1=s_3^2r_1(\lambda_0)$.
\end{prop}

\begin{remark}
Proposition \ref{prop3} is a special case of Proposition \ref{prop4} when $n=1$. Therefore, when $n=m=1$ and $\lambda_0=s_1+s_2$, the solution of RH problem \ref{RH3} is the single breather
$q_1(x,t)$ with the eigenvalue $\lambda_0$ with the normalization constants $s_3^2r_1(\lambda_0)$(see Figure \ref{fig-1}(b)).
\end{remark}

In particular, we take $q_0=1$. When $s_1+s_2\to i$,  we obtain the Peregrine's rational solution (rogue wave) $q_{rw}(x,t)$ of the focusing NLS equation~\cite{rw}:
\bee\no
q_{rw}(x,t)=1-\dfrac{4(4it+1)}{4x^2+16t^2+1}.
\ene

\v \noindent {\bf Case II.} The $n$-breather solution. In this case, we choose $n=m$. Then we arrive at the following Proposition.

\begin{prop}\label{prop4} If $\{\lambda_1,\lambda_2,\cdots,\lambda_n\}$ is a solution to equation $(k-s_1)^n=s_2$, then the solution of the RH problem \ref{RH3} is the $n$-breather solution $q_n(x,t)$ with discrete eigenvalues $\lambda_j,j=1,\cdots,n$ and the normalization constants $c_j=\dfrac{s_3^2r_1(\lambda_j)}{\prod_{k\neq j}(\lambda_j-\lambda_k)},j=1,\cdots,n$.
\end{prop}

\begin{proof}

The boundary of $\Omega_1^*$, which is the complex-conjugate domain of $\Omega_1$, is defined by
\bee\label{propo2-1}
k^*=s_1^*+\left(s_2^*+\dfrac{s_3^2}{(k-s_1)^m-s_2}\right)^{\frac{1}{n}},\quad k\in\partial\Omega_1.
\ene

Substituting Eq.~(\ref{propo2-1}) into Eq.~(\ref{v2}), we obtain
\begin{align}\label{propo2-2}
\begin{aligned}
&\int_{\partial{\Omega_1}}\frac{(\lambda^*-s_1^*)^nr_1(\lambda)e^{-2i\vartheta(\lambda)}}{2\pi i(k-\lambda)}d\lambda
=\sum_{j=1}^n\frac{s_3^2r_1(\lambda_j)e^{-2i\vartheta(\lambda_j)}}{\prod_{k\neq j}(\lambda_j-\lambda_k)(k-\lambda_j)},\v\\
&\int_{\partial{\Omega_1}}\dfrac{\frac{q_0^2q_-^*}{\lambda^{*2}q_-}(\lambda-s_1)^nr_1^*(\lambda)
e^{-2i\vartheta(-\frac{q_0^2}{\lambda^*})}}
{2\pi i(k+\frac{q_0^2}{\lambda^*})}d\lambda^*
=-\sum_{j=1}^n\dfrac{\frac{q_0^2q_-^*}{\lambda_j^{*2}q_-}s_3^2r_1^*(\lambda_j)e^{-2i\vartheta(-\frac{q_0^2}{\lambda_j^*})}}
{\prod_{k\neq j}(\lambda_j^*-\lambda_k^*)(k+\frac{q_0^2}{\lambda_j^*})},\v\\
&\int_{\partial{\Omega_1}}\frac{(\lambda-s_1)^nr_1^*(\lambda)e^{2i\vartheta(\lambda^*)}}{2\pi i(k-\lambda^*)}d\lambda^*
=-\sum_{j=1}^n\frac{s_3^2r_1^*(\lambda_j)e^{2i\vartheta(\lambda_j^*)}}{\prod_{k\neq j}(\lambda_j^*-\lambda_k^*)(k-\lambda_j^*)},\v\\
&\int_{\partial{\Omega_1}}\dfrac{\frac{q_0^2q_-}{\lambda^2q_-^*}(\lambda^*-s_1^*)^nr_1(\lambda)e^{2i\vartheta(-\frac{q_0^2}{\lambda})}}{2\pi i(k+\frac{q_0^2}{\lambda})}d\lambda
=\sum_{j=1}^n\dfrac{\frac{q_0^2q_-}{\lambda_j^2q_-^*}s_3^2r_1(\lambda_j)e^{2i\vartheta(-\frac{q_0^2}{\lambda_j})}}{\prod_{k\neq j}(\lambda_j-\lambda_k)(k+\frac{q_0^2}{\lambda_j})},
\end{aligned}
\end{align}
Then Eqs.~(\ref{v2}) can be rewritten as
\bee
V_2(x,t;k)|_{m=n}
=\begin{cases}
\left[\!\!\begin{array}{cc}
1& 0  \vspace{0.05in}\\
-\sum\limits_{j=1}^n\dfrac{s_3^2r_1(\lambda_j)e^{-2i\vartheta(\lambda_j)}}{\prod_{k\neq j}(\lambda_j-\lambda_k)(k-\lambda_j)}& 1
\end{array}\!\!\right],\quad k\in\Gamma_{1+},\vspace{0.05in}\\
\left[\!\!\begin{array}{cc}
1& 0  \vspace{0.05in}\\
\sum\limits_{j=1}^n\dfrac{\frac{q_0^2q_-^*}{\lambda_j^{*2}q_-}s_3^2r_1^*(\lambda_j)e^{-2i\vartheta(-\frac{q_0^2}{\lambda_j^*})}}
{\prod_{k\neq j}(\lambda_j^*-\lambda_k^*)(k+\frac{q_0^2}{\lambda_j^*})}& 1
\end{array}\!\!\right],\quad k\in\Gamma_{2+},\vspace{0.05in}\\
\left[\!\!\begin{array}{cc}
1& -\sum\limits_{j=1}^n\dfrac{s_3^2r_1^*(\lambda_j)e^{2i\vartheta(\lambda_j^*)}}{\prod_{k\neq j}(\lambda_j^*-\lambda_k^*)(k-\lambda_j^*)}  \vspace{0.05in}\\
0& 1
\end{array}\!\!\right],\quad k\in\Gamma_{1-},\v\\
\left[\!\!\begin{array}{cc}
1&\sum\limits_{j=1}^n\dfrac{\frac{q_0^2q_-}{\lambda_j^2q_-^*}s_3^2r_1(\lambda_j)e^{2i\vartheta(-\frac{q_0^2}{\lambda_j})}}{\prod_{k\neq j}(\lambda_j-\lambda_k)(k+\frac{q_0^2}{\lambda_j})}  \vspace{0.05in}\\
0& 1
\end{array}\!\!\right],\quad k\in\Gamma_{2-}.
\end{cases}
\ene

Then, we consider following transformation:
\bee
\tilde{M}_2(x,t;k)
=\begin{cases}
M_2(x,t;k)\left[\!\!\begin{array}{cc}
1& 0  \vspace{0.05in}\\
\sum\limits_{j=1}^n\dfrac{s_3^2r_1(\lambda_j)e^{-2i\vartheta(\lambda_j)}}{\prod_{k\neq j}(\lambda_j-\lambda_k)(k-\lambda_j)}& 1
\end{array}\!\!\right],\quad k~\mathrm{within}~\Gamma_{1+},\vspace{0.05in}\\
M_2(x,t;k)\left[\!\!\begin{array}{cc}
1& 0  \vspace{0.05in}\\
-\sum\limits_{j=1}^n\dfrac{\frac{q_0^2q_-^*}{\lambda_j^{*2}q_-}s_3^2r_1^*(\lambda_j)e^{-2i\vartheta(-\frac{q_0^2}{\lambda_j^*})}}
{\prod_{k\neq j}(\lambda_j^*-\lambda_k^*)(k+\frac{q_0^2}{\lambda_j^*})}& 1
\end{array}\!\!\right],\quad k~\mathrm{within}~\Gamma_{2+},\vspace{0.05in}\\
M_2(x,t;k)\left[\!\!\begin{array}{cc}
1& -\sum\limits_{j=1}^n\dfrac{s_3^2r_1^*(\lambda_j)e^{2i\vartheta(\lambda_j^*)}}{\prod_{k\neq j}(\lambda_j^*-\lambda_k^*)(k-\lambda_j^*)}  \vspace{0.05in}\\
0& 1
\end{array}\!\!\right],\quad k~\mathrm{within}~\Gamma_{1-},\v\\
M_2(x,t;k)\left[\!\!\begin{array}{cc}
1&\sum\limits_{j=1}^n\dfrac{\frac{q_0^2q_-}{\lambda_j^2q_-^*}s_3^2r_1(\lambda_j)e^{2i\vartheta(-\frac{q_0^2}{\lambda_j})}}{\prod_{k\neq j}(\lambda_j-\lambda_k)(k+\frac{q_0^2}{\lambda_j})}  \vspace{0.05in}\\
0& 1
\end{array}\!\!\right],\quad k~\mathrm{within}~\Gamma_{2-},\v\\
M_2(x,t;k),\quad \mathrm{otherwise}.
\end{cases}
\ene

Through the above transformations, we can obtain the residue condition for the matrix function $\tilde{M}_2(x,t;k)$. $\tilde{M}_2(x,t;k)$ has simple poles at each point in $\{\lambda_j\}_{j=1}^n$ with
\begin{align}
\begin{aligned}
&\mathop{\mathrm{Res}}\limits_{k=\lambda_j}\tilde{M}_2(x,t;k)=\lim\limits_{k\to \lambda_j}\tilde{M}_2(x,t;k)\left[\!\!\begin{array}{cc}
0& 0  \vspace{0.05in}\\
\dfrac{s_3^2r_1(\lambda_j)}{\prod_{k\neq j}(\lambda_j-\lambda_k)}e^{-2i\vartheta(\lambda_j)}& 0
\end{array}\!\!\right],\vspace{0.05in}\\
&\mathop{\mathrm{Res}}\limits_{k=-\frac{q_0^2}{\lambda_j^*}}\tilde{M}_2(x,t;k)=\lim\limits_{k\to -\frac{q_0^2}{\lambda_j^*}}\tilde{M}_2(x,t;k)\left[\!\!\begin{array}{cc}
0& 0  \vspace{0.05in}\\
-\frac{q_0^2q_-^*}{\lambda_j^{*2}q_-}\dfrac{s_3^2r_1^*(\lambda_j)}{\prod_{k\neq j}(\lambda^*_j-\lambda^*_k)}e^{-2i\vartheta(-\frac{q_0^2}{\lambda_j^*})}& 0
\end{array}\!\!\right],\vspace{0.05in}\\
&\mathop{\mathrm{Res}}\limits_{k=\lambda_j^*}\tilde{M}_2(x,t;k)=\lim\limits_{k\to \lambda_j^*}\tilde{M}_2(x,t;k)\left[\!\!\begin{array}{cc}
0& -\dfrac{s_3^2r_1^*(\lambda_j)}{\prod_{k\neq j}(\lambda^*_j-\lambda^*_k)}e^{2i\vartheta(\lambda_j^*)}  \vspace{0.05in}\\
0& 0
\end{array}\!\!\right], \vspace{0.05in}\\
&\mathop{\mathrm{Res}}\limits_{k=-\frac{q_0^2}{\lambda_j}}\tilde{M}_2(x,t;k)=\lim\limits_{k\to -\frac{q_0^2}{\lambda_j}}\tilde{M}_2(x,t;k)\left[\!\!\begin{array}{cc}
0& \frac{q_0^2q_-}{\lambda_j^2q_-^*}\dfrac{s_3^2r_1(\lambda_j)}{\prod_{k\neq j}(\lambda_j-\lambda_k)}e^{2i\vartheta(-\frac{q_0^2}{\lambda_j})}  \vspace{0.05in}\\
0& 0
\end{array}\!\!\right],
\end{aligned}
\end{align}

Therefore, the solution of the RH problem \ref{RH3} is the $n$-breather state $q_n(x,t)$, with discrete eigenvalues $\lambda_j,\, j=1,\cdots,n$ and normalization constants $\dfrac{s_3^2r_1(\lambda_j)}{\prod_{k\neq j}(\lambda_j-\lambda_k)},j=1,\cdots,n$.
\end{proof}

\v \noindent {\bf Case III.} The $n$-soliton solution. Note that the solution of RH problem 3 with $q_0\to 0$ can reduce to the known $n$-soliton solution $q_{n}(x,t)$~\cite{Grava-3}, which can be also directly derived as the limit for $q_{0}\rightarrow 0$ of the obtained $n$-breather solution. Here we choose $n=m$. Then we arrive at the following remark.

\begin{remark} If $\{\lambda_1,\lambda_2,\cdots,\lambda_n\}$ is the solution to equation $(k-s_1)^n=s_2$ and $q_0\to0$, then the solution of the RH problem \ref{RH3} is the $n$-soliton state $q_n(x,t)$ with discrete eigenvalues $\lambda_j,j=1,\cdots,n$ and normalization constants $c_j=\dfrac{s_3^2r_1(\lambda_j)}{\prod_{k\neq j}(\lambda_j-\lambda_k)},j=1,\cdots,n$.
\end{remark}

\subsection{The line domain}

In this section, we address the limit of $N\to\infty$, under the additional assumptions:

\begin{itemize}

 \item {} Poles $\{k_j\}_{j=1}^N$ are sampled from a smooth density function $\rho(k)$ so that $\int_{a}^{-ik_j}\rho(\lambda)d\lambda=\frac{j}{N},j=1,\cdots,N$.

 \item {} The coefficients $\{c_j\}_{j=1}^N$ satisfy
\bee\label{c-j-2}
c_j=\dfrac{i(b-a)r(k_j)}{N\pi},\quad b>a>0,
\ene
where $r(k)$ is a real-valued, continuous and non-vanishing function of $k\in(ia,ib)$, subject to the symmetry relation, $r(k^*)=r(k)=r(-\frac{q_0^2}{k})=r(-\frac{q_0^2}{k^*})$.

\end{itemize}

\begin{prop}\label{le3}
For any open set $A_+(B_+)$ containing the interval $[ia,ib]([-\frac{iq_0^2}{a},-\frac{iq_0^2}{b}])$, and any open set $A_-(B_-)$ containing the interval $[-ib,-ia]([\frac{iq_0^2}{b},\frac{iq_0^2}{a}])$, the following identities hold:
\bee
\lim\limits_{N\to\infty}\sum\limits_{j=1}^{N}\frac{c_je^{-2i\vartheta(k_j)}}{k-k_j}
=\int_{ia}^{ib}\dfrac{r(w)e^{-2i\vartheta(w)}}{\pi(k-w)}dw,
\ene
uniformly for all $\mathbb{C}\setminus A_+$.
\bee
\lim\limits_{N\to\infty}\sum\limits_{j=1}^{N}\dfrac{\frac{q_0^2q_-^*}{k_j^{*2}q_-}c_j^*
e^{-2i\vartheta(-\frac{q_0^2}{k_j^*})}}{k+\frac{q_0^2}{k_j^*}}
=\int_{-\frac{iq_0^2}{a}}^{-\frac{iq_0^2}{b}}\dfrac{q_-^*r(w)e^{-2i\vartheta(w)}}{q_-\pi(k-w)}dw,
\ene
uniformly for all $\mathbb{C}\setminus B_+$.
\bee
\lim\limits_{N\to\infty}\sum\limits_{j=1}^{N}\frac{c_j^*e^{2i\vartheta(k_j^*)}}{k-k_j^*}=
-\int_{-ib}^{-ia}\dfrac{r(w)e^{2i\vartheta(w)}}{\pi(k-w)}dw,
\ene
uniformly for all $\mathbb{C}\setminus A_-$.
\bee
\lim\limits_{N\to\infty}\sum\limits_{j=1}^{N}\dfrac{\frac{q_0^2q_-}{k_j^2q_-^*}c_j
e^{2i\vartheta(-\frac{q_0^2}{k_j})}}{k+\frac{q_0^2}{k_j}}
=-\int_{\frac{iq_0^2}{b}}^{\frac{iq_0^2}{a}}\dfrac{q_-r(w)e^{2i\vartheta(w)}}{q^*_-\pi(k-w)}dw,
\ene
uniformly for all $\mathbb{C}\setminus B_-$. The open intervals $(ia,ib),(-\frac{iq_0^2}{a},-\frac{iq_0^2}{b}),(\frac{iq_0^2}{b},\frac{iq_0^2}{a})$ and $(-ib,-ia)$ are both oriented upwards.

\end{prop}

\begin{proof}
Using Eq.~(\ref{c-j-2}), we have
\begin{align}\no
\begin{aligned}
\lim\limits_{N\to\infty}\sum\limits_{j=1}^{N}\frac{c_je^{-2i\vartheta(k_j)}}{k-k_j}
=\lim\limits_{N\to\infty}\sum\limits_{j=1}^{N}\frac{i(b-a)}{N}\frac{r(k_j)e^{-2i\vartheta(k_j)}}{\pi(k-k_j)}
=\int_{ia}^{ib}\dfrac{r(w)e^{-2i\vartheta(w)}}{\pi(k-w)}dw,
\end{aligned}
\end{align}
and
\begin{align}\no
\begin{aligned}
\lim\limits_{N\to\infty}\sum\limits_{j=1}^{N}\dfrac{\frac{q_0^2q_-^*}{k_j^{*2}q_-}c_j^*
e^{-2i\vartheta(-\frac{q_0^2}{k_j^*})}}{k+\frac{q_0^2}{k_j^*}}
=&\lim\limits_{N\to\infty}\sum\limits_{j=1}^{N}-\frac{i(b-a)}{N}\dfrac{\frac{q_0^2q_-^*}{k_j^{*2}q_-}r(k_j)
e^{-2i\vartheta(-\frac{q_0^2}{k_j^*})}}{\pi(k+\frac{q_0^2}{k_j^*})}\v\\
=&-\int_{-ib}^{-ia}\dfrac{\frac{q_0^2q_-^*}{\lambda^{2}q_-}r(\lambda)
e^{-2i\vartheta(-\frac{q_0^2}{\lambda})}}{\pi(k+\frac{q_0^2}{\lambda})}d\lambda\v\\
=&\int_{-\frac{iq_0^2}{a}}^{-\frac{iq_0^2}{b}}\dfrac{q_-^*r(w)e^{-2i\vartheta(w)}}{q_-\pi(k-w)}dw,
\end{aligned}
\end{align}

Thus the proof is completed.
\end{proof}

At $N\to\infty$, according to Proposition~\ref{le3}, the jump matrix $V_1(x,t;k)$, defined by Eq.~(\ref{V1-1}), can be rewritten as:
\bee\label{V1-12}
V_1(x,t;k)|_{N\to\infty}
=\begin{cases}
\left[\!\!\begin{array}{cc}
1& 0  \vspace{0.05in}\\
-\d\int_{ia}^{ib}\dfrac{r(w)e^{-2i\vartheta(w)}}{\pi(k-w)}dw& 1
\end{array}\!\!\right],\quad k\in\Gamma_{1+},\vspace{0.05in}\\
\left[\!\!\begin{array}{cc}
1& 0  \vspace{0.05in}\\
\d\int_{-\frac{iq_0^2}{a}}^{-\frac{iq_0^2}{b}}\dfrac{q_-^*r(w)e^{-2i\vartheta(w)}}{q_-\pi(k-w)}dw& 1
\end{array}\!\!\right],\quad k\in\Gamma_{2+},\vspace{0.05in}\\
\left[\!\!\begin{array}{cc}
1& \d\int_{-ib}^{-ia}\dfrac{r(w)e^{2i\vartheta(w)}}{\pi(k-w)}dw  \vspace{0.05in}\\
0& 1
\end{array}\!\!\right],\quad k\in\Gamma_{1-},\v\\
\left[\!\!\begin{array}{cc}
1&-\d\int_{\frac{iq_0^2}{b}}^{\frac{iq_0^2}{a}}\dfrac{q_-r(w)e^{2i\vartheta(w)}}{q^*_-\pi(k-w)}dw  \vspace{0.05in}\\
0& 1
\end{array}\!\!\right],\quad k\in\Gamma_{2-}.
\end{cases}
\ene

Apply the following transformation:
\bee
M_3(x,t;k)
=\begin{cases}
M_1(x,t;k)\left[\!\!\begin{array}{cc}
1& 0  \vspace{0.05in}\\
\d\int_{ia}^{ib}\dfrac{r(w)e^{-2i\vartheta(w)}}{\pi(k-w)}dw& 1
\end{array}\!\!\right],\quad k~\mathrm{within}~\Gamma_{1+},\vspace{0.05in}\\
M_1(x,t;k)\left[\!\!\begin{array}{cc}
1& 0  \vspace{0.05in}\\
-\d\int_{-\frac{iq_0^2}{a}}^{-\frac{iq_0^2}{b}}\dfrac{q_-^*r(w)e^{-2i\vartheta(w)}}{q_-\pi(k-w)}dw& 1
\end{array}\!\!\right],\quad k~\mathrm{within}~\Gamma_{2+},\v\\
M_1(x,t;k)\left[\!\!\begin{array}{cc}
1& \d\int_{-ib}^{-ia}\dfrac{r(w)e^{2i\vartheta(w)}}{\pi(k-w)}dw  \vspace{0.05in}\\
0& 1
\end{array}\!\!\right],\quad k~\mathrm{within}~\Gamma_{1-},\v\\
M_1(x,t;k)\left[\!\!\begin{array}{cc}
1& -\d\int_{\frac{iq_0^2}{b}}^{\frac{iq_0^2}{a}}\dfrac{q_-r(w)e^{2i\vartheta(w)}}{q^*_-\pi(k-w)}dw  \vspace{0.05in}\\
0& 1
\end{array}\!\!\right],\quad k~\mathrm{within}~\Gamma_{2-},\v\\
M_1(x,t;k),\quad \mathrm{otherwise}.
\end{cases}
\ene

Then the matrix function $M_3(x,t;k)$ satisfies the following RH problem:

\begin{RH}\label{RH4}
Find a $2\times 2$ matrix function $M_3(x,t;k)$ that satisfies the particular conditions:

\begin{itemize}

 \item {} Analyticity: $M_3(x,t;k)$ is analytic in $\mathbb{C}\setminus((ia,ib)\cup(-\frac{iq_0^2}{a},-\frac{iq_0^2}{b})\cup(-ib,-ia)\cup(\frac{iq_0^2}{b},\frac{iq_0^2}{a}))$ and takes continuous boundary values on $(ia,ib)\cup(-\frac{iq_0^2}{a},-\frac{iq_0^2}{b})\cup(-ib,-ia)\cup(\frac{iq_0^2}{b},\frac{iq_0^2}{a})$(Directions of these open intervals are all facing upwards).

 \item {} The jump condition: The boundary values on the jump contour $(ia,ib)\cup(-\frac{iq_0^2}{a},-\frac{iq_0^2}{b})\cup(-ib,-ia)\cup(\frac{iq_0^2}{b},\frac{iq_0^2}{a})$ are defined as
 \bee\no
M_{3+}(x,t;k)=M_{3-}(x,t;k)V_3(x,t;k),\,\, \lambda\in(ia,ib)\cup(-\frac{iq_0^2}{a},-\frac{iq_0^2}{b})\cup(-ib,-ia)\cup(\frac{iq_0^2}{b},\frac{iq_0^2}{a}),
\ene
where
\bee
V_3(x,t;k)
=\begin{cases}
\left[\!\!\begin{array}{cc}
1& 0  \vspace{0.05in}\\
-2ir(k)e^{-2i\vartheta(k)}& 1
\end{array}\!\!\right],\quad k\in(ia,ib),\vspace{0.05in}\\
\left[\!\!\begin{array}{cc}
1& 0  \vspace{0.05in}\\
\dfrac{2iq_-^*r(k)e^{-2i\vartheta(k)}}{q_-}& 1
\end{array}\!\!\right],\quad k\in(-\frac{iq_0^2}{a},-\frac{iq_0^2}{b}),\vspace{0.05in}\\
\left[\!\!\begin{array}{cc}
1& -2ir(k)e^{2i\vartheta(k)}  \vspace{0.05in}\\
0& 1
\end{array}\!\!\right],\quad k\in(-ib,-ia),\v\\
\left[\!\!\begin{array}{cc}
1&  \dfrac{2iq_-r(k)e^{2i\vartheta(k)}}{q^*_-}  \vspace{0.05in}\\
0& 1
\end{array}\!\!\right],\quad k\in(\frac{iq_0^2}{b},\frac{iq_0^2}{a}).
\end{cases}
\ene

 \item {} The normalization:
\bee
 M_3(x, t; k)=\left\{\begin{array}{ll}
    \mathbb{I}_2+O\left(1/k\right)  & k\to\infty, \v\\
    \dfrac{i}{k}\,\sigma_3\,Q_-+O\left(1\right), & k\to 0.
\end{array}\right.
\ene

\end{itemize}
\end{RH}

 \begin{proof}
Using the Plemelj formula, we have
\bee
\begin{array}{rl}
M_{3+}(x,t;k)=&M_{1+}(x,t;k)\left[\!\!
\begin{array}{cc}
1& 0  \vspace{0.05in}\\
\d\int_{ia}^{ib}\dfrac{r(w)e^{-2i\vartheta(w)}}{\pi(k_+-w)}dw& 1
\end{array}\!\!\right] \v\\
=&M_{1-}(x,t;k)\left[\!\!\begin{array}{cc}
1& 0  \vspace{0.05in}\\
\d\int_{ia}^{ib}\dfrac{r(w)e^{-2i\vartheta(w)}}{\pi(k_--w)}dw& 1
\end{array}\!\!\right]\left[\!\!\begin{array}{cc}
1& 0  \vspace{0.05in}\\
-2ir(k)e^{-2i\vartheta(k)}& 1
\end{array}\!\!\right],\quad k\in(ia,ib),
\end{array}
\ene
Then, we have
\bee
V_3(x,t;k)
\left[\!\!\begin{array}{cc}
1& 0  \vspace{0.05in}\\
-2ir(k)e^{-2i\vartheta(k)}& 1
\end{array}\!\!\right],\quad k\in(ia,ib).
\ene
Using the same method, we can prove other cases as well.
\end{proof}

According to Eq.~(\ref{fanyan-1}), we recover $q(x,t)$ by the following formula:
\bee\label{fanyan-2g}
q(x,t)=-i\lim\limits_{k\rightarrow\infty}\left(k M_{3}(x,t;k)\right)_{12}.
\ene

Then the RH problem \ref{RH4} for the matrix function $M_3(x,t;k)$ represents the breather gas.

\subsection{The elliptic domain}

In this section, we consider the limit of $N\to\infty$, under the additional assumptions:
\begin{itemize}

 \item {} Discrete eigenvalues $k_j,j=1,\cdots,N$ with normalization constants $c_j,j=1,\cdots,N$ fill a uniformly compact domain $\Omega_2$ of the complex upper half plane $\mathbb{C}_+$, that is,
\bee
\Omega_2:=\left\{k\in\mathbb{C}|~\frac{{\rm Re}(k)^{2}}{b_{2}^{2}}+\frac{(2{\rm Im}(k)-a_{1}-a_{2})^{2}}{4b_{1}^{2}}<1\right\},\quad \Omega_2\subset D_+,
\ene
where $ia_1$ and $ia_2$\, ($a_2>a_1$) are the focal points of the ellipse $\partial\Omega_2$,  $b_1=\sqrt{b_2^2+(\frac{a_2-a_1}{2})^2}$, and $b_2$
is sufficiently small so that $\Omega_2$ lies in domain $D_+$.

 \item {} The normalization constants $c_j,j=1,\cdots,N$ have the following form:
\bee
c_j=\frac{|\Omega_2|r_1(k_j)}{N\pi}.
\ene
where $|\Omega_2|$ means the area of the domain $\Omega_2$ and $r_1(k)$ is an analytic functions in domain $\Omega_2$, subject to the symmetry condition $r_1^*(k)=r_1(k^*)$.

\end{itemize}

According to RH problems~\ref{le1} and \ref{le2}, we arrive at the following RH problem $M_4(x,t;k):=\lim\limits_{N\to\infty}M_1(x,t;k)$:

\begin{RH}\label{RH5}
Find a $2\times 2$ matrix function $M_4(x,t;k)$ that meets the particular conditions:

\begin{itemize}

 \item {} Analyticity: $M_4(x,t;k)$ is analytic in $\mathbb{C}\setminus(\Gamma_{1\pm}\cup\Gamma_{2\pm})$ and takes continuous boundary values on $\Gamma_{1\pm}\cup\Gamma_{2\pm}$.

 \item {} The jump condition: The boundary values on the jump contour $\Gamma_{1+}\cup\Gamma_{1-}$ are defined as
 \bee
M_{4+}(x,t;k)=M_{4-}(x,t;k)V_4(x,t;k),\quad \lambda\in\Gamma_{1\pm}\cup\Gamma_{2\pm},
\ene
where
\bee\label{v4}
V_4(x,t;k)
=\begin{cases}
\left[\!\!\begin{array}{cc}
1& 0  \vspace{0.05in}\\
- \d\int_{\partial{\Omega_2}}\frac{\lambda^*r_1(\lambda)e^{-2i\vartheta(\lambda)}}{2\pi i(k-\lambda)}d\lambda& 1
\end{array}\!\!\right],\quad k\in\Gamma_{1+},\vspace{0.05in}\\
\left[\!\!\begin{array}{cc}
1& 0  \vspace{0.05in}\\
-\d\int_{\partial{\Omega_2}}\dfrac{\frac{q_0^2q_-^*}{\lambda^{*2}q_-}\lambda r_1^*(\lambda)e^{-2i\vartheta(-\frac{q_0^2}{\lambda^*})}}
{2\pi i(k+\frac{q_0^2}{\lambda^*})}d\lambda^*& 1
\end{array}\!\!\right],\quad k\in\Gamma_{2+},\vspace{0.05in}\\
\left[\!\!\begin{array}{cc}
1& \d\int_{\partial{\Omega_2}}\frac{\lambda r_1^*(\lambda)e^{2i\vartheta(\lambda^*)}}{2\pi i(k-\lambda^*)}d\lambda^*  \vspace{0.05in}\\
0& 1
\end{array}\!\!\right],\quad k\in\Gamma_{1-},\v\\
\left[\!\!\begin{array}{cc}
1& \d\int_{\partial{\Omega_2}}\dfrac{\frac{q_0^2q_-}{\lambda^2q_-^*}\lambda^*r_1(\lambda)e^{2i\vartheta(-\frac{q_0^2}{\lambda})}}{2\pi i(k+\frac{q_0^2}{\lambda})}d\lambda  \vspace{0.05in}\\
0& 1
\end{array}\!\!\right],\quad k\in\Gamma_{2-}.
\end{cases}
\ene

 \item {} The normalization:
\bee
 M_4(x, t; k)=\left\{\begin{array}{ll}
    \mathbb{I}_2+O\left(1/k\right)  & k\to\infty, \v\\
    \dfrac{i}{k}\,\sigma_3\,Q_-+O\left(1\right), & k\to 0.
\end{array}\right.
\ene

\end{itemize}
\end{RH}

\begin{prop}\label{le4} The following identities hold:
\begin{align}\label{le4-1}
\begin{aligned}
&\int_{\partial{\Omega_1}}\frac{\lambda^*r_1(\lambda)e^{-2i\vartheta(\lambda)}}{2\pi i(k-\lambda)}d\lambda=\int_{ia_1}^{ia_2}\frac{\Delta F(k)r_1(\lambda)e^{-2i\vartheta(\lambda)}}{2\pi i(k-\lambda)}d\lambda,
\v\\
&\int_{\partial{\Omega_1}}\dfrac{\frac{q_0^2q_-^*}{\lambda^{*2}q_-}\lambda r_1^*(\lambda)e^{-2i\vartheta(-\frac{q_0^2}{\lambda^*})}}
{2\pi i(k+\frac{q_0^2}{\lambda^*})}d\lambda^*=\int_{-\frac{iq_0^2}{a}}^{-\frac{iq_0^2}{b}}\dfrac{\frac{q_-^*}{q_-}\Delta F^*(-\frac{q_0^2}{w^*})r_1(w)e^{-2i\vartheta(w)}}{2\pi i(k-w)}dw,\v\\
&\int_{\partial{\Omega_1}}\frac{\lambda r_1^*(\lambda)e^{2i\vartheta(\lambda^*)}}{2\pi i(k-\lambda^*)}d\lambda^*
=-\int_{-ia_2}^{-ia_1}\frac{\Delta F^*(k)r_1(\lambda)e^{2i\vartheta(\lambda)}}{2\pi i(k-\lambda)}d\lambda,\v\\
&\int_{\partial{\Omega_1}}\dfrac{\frac{q_0^2q_-}{\lambda^2q_-^*}\lambda^*r_1(\lambda)e^{2i\vartheta(-\frac{q_0^2}{\lambda})}}{2\pi i(k+\frac{q_0^2}{\lambda})}d\lambda=-\int_{\frac{iq_0^2}{b}}^{\frac{iq_0^2}{a}}\dfrac{\frac{q_-}{q_-^*}\Delta F(-\frac{q_0^2}{w})r_1(w)e^{2i\vartheta(w)}}{2\pi i(k-w)}dw,
\end{aligned}
\end{align}
where $\Delta F(k)=F_+(k)-F_-(k)$, and the function $F(k)$ is
analytic in complex plane $\mathbb{C}$ away from the segment $%
[ia_{1},ia_{2}] $, with boundary values $F_{\pm }(k)$.
\end{prop}

\begin{proof}

The boundary of the complex-conjugate domain $\Omega _{2}^{\ast }$ of $%
\Omega _{2}$ is described by
\bee\label{le4-2}
k^*=\left(1-\frac{8b_1^2}{(a_2-a_1)^2}\right)\left(k-\frac{i(a_1+a_2)}{2}\right)+\frac{8b_1b_2}{(a_2-a_1)^2}F(k)-\frac{i(a_1+a_2)}{2},\quad k\in\partial\Omega_2,
\ene
where $F(k)^2=(k-ia_1)(k-ia_2)$.
Using Eq.~(\ref{le4-2}), one can obtain Eq.~(\ref{le4-1}).
\end{proof}

According to Proposition \ref{le4}, we find that the RH problem \ref{RH5} for matrix function $M_4(x,t;k)$ is equivalent to the one in case of the line domain.

\begin{remark}
As $q_0\to0$ and $t=0$, the limiting initial data is step-like oscillatory with the elliptic travelling wave of
type $dn(x)$ as $x\to-\infty$ and exponentially going to zero as $x\to+\infty$. For detailed parameters, refer to references~\cite{Grava-3,Bertola-pa}.
\end{remark}

\section{Conclusions and discussions}

Based on the IST and RH problems, we have investigated a breather gas represented by
the $N_{\infty}$-breather solution of the focusing NLS equation with nonzero BCs. In terms of scattering data of IST, the $N$-breather
solutions are based on the set of discrete eigenvalues $K\equiv \{k_{j},-\frac{%
q_{0}^{2}}{k_{j}^{\ast }},k_{j}^{\ast },-\frac{q_{0}^{2}}{k_{j}}\}_{j=1}^{N}$%
with normalization constants.
$\{c_{j},-\frac{%
q_-^{*2}c_{j}^{\ast}}{k_{j}^{\ast 2}},-c_{j}^{\ast },%
\frac{q_-^{2}c_{j}}{k_{j}^{2}}\}_{j=1}^{N}$.
By concentrating the set of $\{k_{j}\}_{j=1}^{N}$
in different domains, we have
produced different types of breather gases which coagulate into thefollowing effective forms: i) The concentration domain in the form of a disk
condenses the gas into the single-breather solution with the spectral
eigenvalue located at the disk's center; ii) The quadrater domain with $m=n$
\ and $q\rightarrow 0$ leads to the coagulation of the gas into the $n$%
-breather solution. These are examples of the
breather-gas shielding. The discrete spectra concentrated in line domains
imply solving the corresponding RH problems. The case of
the discrete spectra lying on an ellipse is tantamount to the case of the
line domain. When discrete spectra are uniformly distributed within a specified region, the interaction among breathers manifests itself in the form of $n$-breathers, where parameter $n$ is correlated with the region in question. The phenomenon of breather shielding can explain the distribution of the breathers when the discrete spectrum is densely distributed. The methodology presented here can be extended to other integrable equations and can also be employed to investigate the asymptotic behavior of breathers in different regions. For the phenomenon of breather shielding, we have developed here only the analytical framework. Verification of the findings by means of numerical methods is a subject for a separate work.

The breather-gas shielding predicted by the present analysis can be
observed in fiber optics, BEC, and other physical realizations of the NLS.
The approach developed in this work can be extended to other
integrable models -- first of all, those based on ZS-type spectral problems.
A
challenging possibility is to extend the analysis of the shielding
phenomenology to quantized NLS fields, in which breather states feature
specific fluctuations \cite{Yurovsky}.

\v \noindent {\bf Acknowledgments} \, We thank Prof. Marco Bertola (Concordia University, Canada) for valuable suggestions and discussions. This work of W.W. was supported by the National Natural Science Foundation of China (No.12501327). The work of Z.Y. was supported by the National Natural Science Foundation of China (No. 12471242). The work of B.A.M. is supported, in part, by Israel Science Foundation (No. 1695/22). The work of G.Z. was supported by National Key R\&D Program of China under Grant 2024YFA1013101 and the National Natural Science Foundation of China (No. 12201615).

\v \noindent {\bf Data availability} No new data were created or analyzed in this study.

\section*{Declarations}

 \noindent {\bf Conflict of interest} On behalf of all authors, the corresponding author states that there is no conflict of
interest.

\end{document}